\documentclass[a4paper,superscriptaddress,showkeys]{revtex4}

\usepackage{graphicx}
\usepackage{caption}
\usepackage{graphicx}
\usepackage{subfigure}
\usepackage{indentfirst}
\usepackage{amsfonts}
\usepackage{amsmath}
\usepackage{amsthm}
\usepackage{amssymb}
\usepackage[T1]{fontenc}
\newcommand{\ket}[1]{\left| #1 \right\rangle}
\newcommand{\bra}[1]{\left\langle #1 \right|}

\usepackage{color}

\newcommand{\eg}{{\textit{e.g.}}, }
\newcommand{\ie}{{\textit{i.e.}}, }
\newtheorem{theorem}{Theorem}
\newtheorem{proposition}{Proposition}

\newtheorem{remark}{Remark}

\begin{document}

\title{Self-catalytic conversion of pure quantum states}

\author{Cristhiano  Duarte}
\email{cristhiano@mat.ufmg.br}
\affiliation{Departamento de Matem\'{a}tica, Instituto de
Ci\^{e}ncias Exatas, Universidade Federal de Juiz de Fora, CEP 36036-330, Juiz de Fora, Minas
Gerais, Brazil.}
\affiliation{Departamento de Matem\'{a}tica, Instituto
de Ci\^{e}ncias Exatas, Universidade Federal de Minas Gerais, CP
702, CEP 30123-970, Belo Horizonte, Minas Gerais, Brazil.}
\author{Raphael C. Drumond}
\email{raphael@mat.ufmg.br}
\affiliation{Departamento de Matem\'{a}tica, Instituto
de Ci\^{e}ncias Exatas, Universidade Federal de Minas Gerais, CP
702, CEP 30123-970, Belo Horizonte, Minas Gerais, Brazil.}
\author{Marcelo  Terra Cunha}
\email{tcunha@ime.unicamp.br}
\affiliation{Departamento de Matem\'{a}tica, Instituto
de Ci\^{e}ncias Exatas, Universidade Federal de Minas Gerais, CP
702, CEP 30123-970, Belo Horizonte, Minas Gerais, Brazil.}
\affiliation{Departamento de Matem\'{a}tica Aplicada, Instituto
de Matem\'atica, Estat\'istica e Computa\c c\~ao Cient\'ifica, Universidade Estadual de Campinas, Cidade Universit\'aria Zeferino Vaz,
CEP 13083-970, Campinas, S\~ao Paulo, Brazil.}
\date{\today}

\begin{abstract}
Conversion of entangled states under (Stochastic) Local Operations and Classical Communication admits the phenomenon of catalysis.
Here we  explore the possibility of {a copy of the initial state itself} to perform
as a catalyst, which we call a self-catalytic process. 
We show explicit examples of self-catalysis.
Necessary and sufficient conditions for the phenomenon to take
place are discussed.
We numerically estimate how frequent it is and we show that increasing the 
number of copies used as catalyst can increases the probability of 
conversion, but do not make the process deterministic. 
By the end we conjecture that under LOCC the probability of finding a 
self-catalytic 
reaction \emph{do not} increase monotonically with the dimensions whereas 
SLOCC, \emph{does} increase.  
\end{abstract}

\keywords{Catalysis, Convertibility, LOCC, SLOCC, Self-Catalisys.}

\maketitle

\section{Introduction}

The conversion between bipartite or multipartite quantum states through 
local operations is a
central concept of entanglement theory. 
For instance, it is the criterion used to classify
entanglement, namely, two states have equivalent entanglement if they can be converted to each
other~\cite{nielsenchuang}. 
Moreover, it is usual and natural to consider a state more entangled than other when the first can 
{\emph{access}} the former by the allowed transformations \cite{nielsen}.
Two important sets of allowed transformations are the deterministic
{\emph{local operations and classical communication}} (LOCC) and its stochastic 
version (SLOCC), where the transformation only needs to succeed with positive 
probability \cite{Andreas}. 
Although this hierarchisation is relatively simple for bipartite pure 
states~\cite{nielsen, horodeckis}, it is
quite involved if mixed or multipartite states are considered~\cite{multipartite1, multipartite2, 
EisertWilkens}.

The problem of convertibility of bipartite pure states has been solved by Nielsen\cite{nielsen}, using the concept 
of majorisation \cite{nielsen2}. 
However, Jonathan and Plenio discovered a surprising effect \cite{daniel}. 
They have shown the existence of pairs of states
which are not directly inter-convertible but such that their conversion is possible if another (necessarily
entangled) state is attached to them. 
That is, the pair of states
have incomparable entanglement but they may become ordered if an extra system is attached
to the original ones. 
Such extra state that makes a transformation possible, without being consumed, is called a {\emph{catalyst}}. 
More recently, the problem of convertibility has received experimental attention \cite{exp} and has 
also been connected to basic results in thermodynamics \cite{brandao}  and phase transitions 
\cite{helena}.

In this paper we explore the following question: can one of the states of a non
inter-convertible pair be used as a catalyst? 
We answer this question positively, providing explicit examples. 
We also explore how common such processes are for low dimensional systems.
An interesting situation is when a state $\ket{\Psi}$ is not able of self-catalysing a transformation, but a number of copies, $\ket{\Psi}^{\otimes n}$, is. 
We exhibit examples where more than one copy is required and study how augmenting the number of copies
can increase the probability of conversion.  

In section~\ref{catalisys} we review basic notions of state conversion and catalysis under LOCC. 
In section~\ref{selfie} we show explicit examples of self-catalytic processes,
explore, through numerical analysis, how frequent they are under random samples of incomparable
pairs and how it depends on the size of the systems. 
Section~\ref{catslocc} review the notions of probabilistic catalysis under SLOCC,
while the natural questions of self-catalysis under SLOCC are discussed in section~\ref{sefslocc}. 
We close the main text with final remarks
and further problems in section~\ref{concl}. 


\section{Catalisys}\label{catalisys}

We say that a bipartite state $\ket{\alpha}\in H_{A}\otimes H_{B}$ \emph{access} a state
$\ket{\beta}\in H'_{A}\otimes H'_{B}$, if there is some LOCC operation, represented by a completely
positive trace preserving (CPTP) map
$\Lambda$, such that $\Lambda(\ket{\alpha}\bra{\alpha})=\ket{\beta}\bra{\beta}$, where
$H_{A},H'_{A},H_{B}$ and $H'_{B}$ are finite dimensional Hilbert spaces and $\Lambda$ maps density
operators acting on $H_{A}\otimes H_{B}$ to density operators acting on 
$H'_{A}\otimes H'_{B}$. In such
case, we write 
$\ket{\alpha}\rightarrow \ket{\beta}$.
If there is no LOCC operation able to convert $\ket{\alpha}$ to $\ket{\beta}$, we shall write $\ket{\alpha} \nrightarrow \ket{\beta}$.

For example, for a pair of qubits, the Bell state $\ket{\Phi^{+}}=(\ket{00}+\ket{11})/\sqrt{2}$ can
access any two-qubit pure state.
Indeed, naming our qubits $A$ and $B$, and 
writing the target state as $\ket{\beta} = a\ket{00} + b\ket{11}$, we can, 
for example, make $A$ unitarily interact
with an auxiliary qubit $A'$, so that $\ket{0_{A}0_{A'}}\mapsto
a\ket{0_{A}0_{A'}}+b\ket{1_{A}1_{A'}}$ and $\ket{1_{A}0_{A'}}\mapsto
a\ket{0_{A}1_{A'}}+b\ket{1_{A}0_{A'}}$. Then,
$[(\ket{0_{A}0_{B}}+\ket{1_{A}1_{B}})/\sqrt{2}]\ket{0_{A'}}\mapsto
(1/\sqrt{2})(a\ket{0_{A}0_{B}}+b\ket{1_{A}1_{B}})\ket{0_{A'}}+(1/\sqrt{2})(a\ket{0_{A}1_{B}}+b\ket{
1_{A}0_{B}} )\ket{ 1_{A'} }$. The lab with qubit $A$ can make a measurement 
on the computational basis of the auxiliary $A'$
and send the result to the lab holding qubit $B$. If the result is $0$, they already share the
desired state, while the result being $1$, a NOT operation,
$\ket{0_{B}}\mapsto\ket{1_{B}},\ket{1_{B}}\mapsto\ket{0_{B}}$, can be applied to qubit $B$ to also
leave the system $AB$ in the desired state.
This result generalises for a pair of qu$d$its in the state $\displaystyle{\ket{\Phi^+_d} = \frac{1}{\sqrt{d}}\sum_{i=0}^{d-1} \ket{ii}}$, which justifies $\ket{\Phi^+_d}$ to be called a maximally entangled state.  

Nielsen, in the seminal paper~\cite{nielsen}, provided a simple necessary and sufficient condition
for determining whether a general bipartite state $\ket{\alpha}$ can access a state
$\ket{\beta}$ in terms of their corresponding Schmidt vectors~\cite{nielsenchuang}:
\begin{theorem}[Nielsen Criterion]
\label{thm-Nielsen}
If
$\vec{\alpha}=(\alpha_{1},...,\alpha_{n})$ and $\vec{\beta}=(\beta_{1},...,\beta_{m})$ are the
(non-increasing) ordered Schmidt
vectors of $\ket{\alpha}$ and $\ket{\beta}$, respectively, we have $\ket{\alpha}\rightarrow
\ket{\beta}$ if, and only if,
\begin{equation}\label{majoracao}
\sum_{l=1}^{k}\alpha_{l}\leq\sum_{l=1}^{k}\beta_{l}
\end{equation} 
for $1\leq k\leq \min\{n,m\}$. 
\end{theorem}
By an ordered Schmidt vector 
$\vec{\lambda} = (\lambda_{1},\ldots,\lambda_{n})$, we mean
that $\lambda_{1}\geq \lambda_{2}\geq \ldots \geq\lambda_{n}$. 
In this work, we shall always assume that
Schmidt vectors are ordered. When vectors $\vec{\alpha}$ and $\vec{\beta}$
satisfy conditions~\eqref{majoracao} we say that $\vec{\alpha}$ is majorised by $\vec{\beta}$ and write
$\vec{\alpha} \preceq \vec{\beta}$.

Revisiting the example above, we have $(1/2,1/2)$ for the ordered Schmidt vector of $\ket{\Phi^{+}}$
and $(|a|^{2},|b|^{2})$ for $a\ket{00}+b\ket{11}$, assuming $|a|\geq |b|$, it is straightforward to
apply the criterion and verify that $\ket{\Phi^{+}}\rightarrow [a\ket{00}+b\ket{11}]$.

A consequence of the criterion is the existence of pair of states $\ket{\alpha}$ and $\ket{\beta}$
such that $\ket{\alpha}\nrightarrow\ket{\beta}$ and $\ket{\beta}\nrightarrow\ket{\alpha}$. 
Actually, the only case where such an order is total is for two qubits.
For instance, consider a state of two  four-level systems with Schmidt vector $(0.4,0.4,0.1,0.1)$ and a
two-qutrit state with Schmidt vector $(0.5,0.25,0.25)$. 
Indeed, 
\begin{subequations}
\begin{align}
0.4=\alpha_{1}&< \beta_{1}=0.5,\\
0.4+0.4=\alpha_{1}+\alpha_{2}&>\beta_{1}+\beta_{2}=0.5+0.25.
\end{align}
\end{subequations}

In Ref.~\cite{daniel}, the authors surprisingly show that it is possible to circumvent
such non-accessibility between these states by making the parts to share an entangled state $\ket{\kappa}$  such that $\ket{\alpha}\otimes \ket{\kappa}\rightarrow \ket{\beta}\otimes \ket{\kappa}$
(see Figure~\ref{cris1}). Since the state $\ket{\kappa}$ allows for a previously
forbidden
conversion, but at the end of the process it remains unaltered, it is called a \emph{catalyst}. 
{In this sense, we say that $\ket{\alpha}$ \emph{$\vec{\kappa}$-access} $\ket{\beta}$ 
when $\ket{\alpha} \nleftrightarrow \ket{\beta}$, but $\ket{\alpha} \otimes \ket{\kappa} 
\rightarrow \ket{\beta} \otimes \ket{\kappa}$}.
In this specific example, the two-qubit state $\ket{\kappa}$ with Schmidt vector
$\vec{\kappa}=(0.6,0.4)$ is a catalyst.

\begin{figure}[h] 
\includegraphics[scale=0.3]{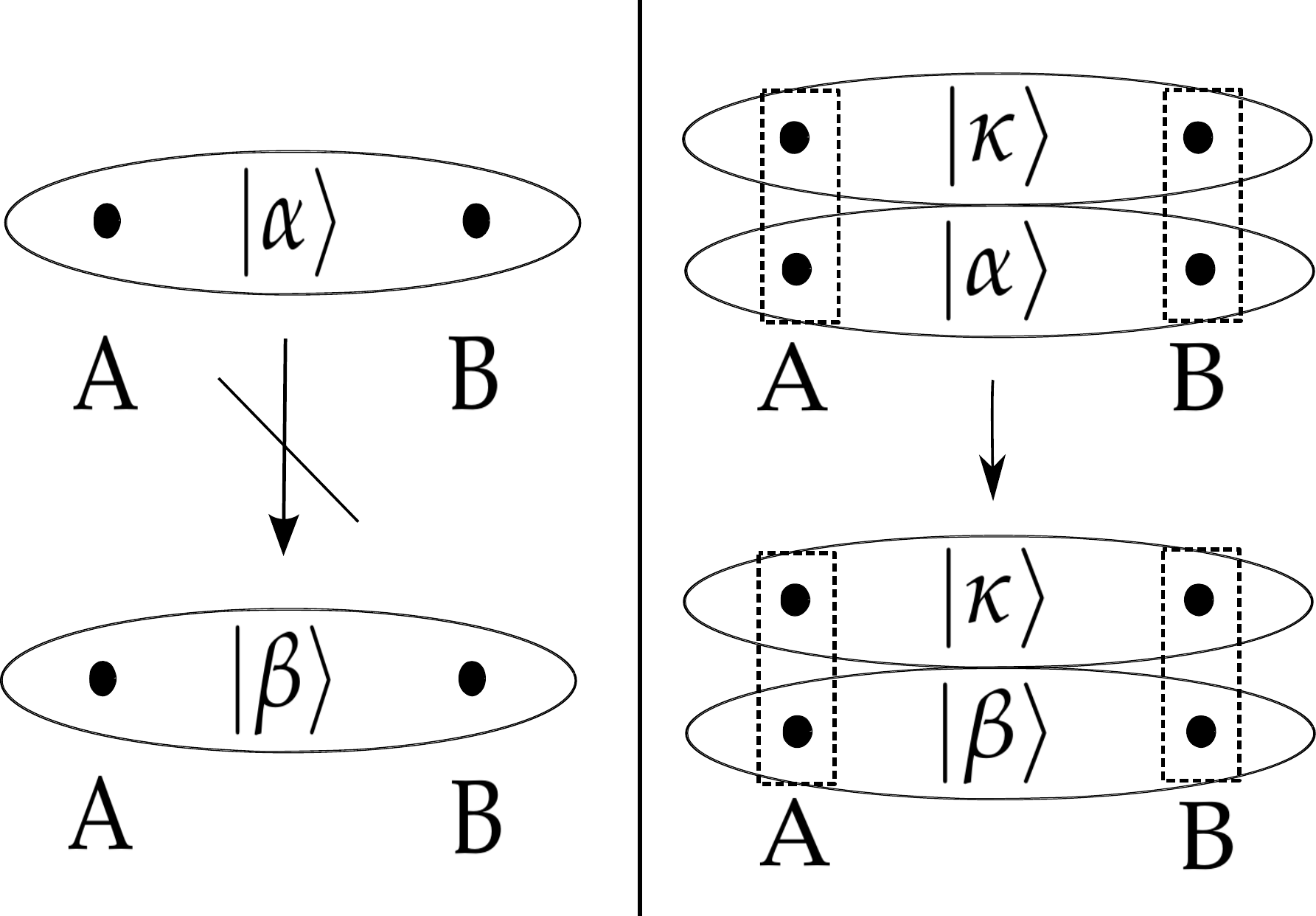}
\caption{Although the conversion $\ket{\alpha}\rightarrow \ket{\beta}$ is not allowed under LOCC,
the state $\ket{\kappa}$ can be used (but not consumed) to make it viable.}\label{cris1}
\end{figure}

Nielsen's criterion ensures that all information about the (possibility of) conversion is 
contained in the Schmidt vectors. 
Therefore, we can explore our knowledge regarding probability 
vectors to provide more examples of catalysts. 
If $\vec{\alpha} = (0.5,0.4,0.05,0.05)$, $\vec{\beta} =(0.7,0.15,0.15)$,  $\vec{\kappa_1}=(0.7,0.3)$, and $\vec{\kappa_2}=(0.75, 0.25)$,
we obtain:
\begin{subequations}
\begin{align}
\vec{\alpha} &\nleftrightarrow \vec{\beta} \\
\vec{\alpha}  \otimes \vec{\kappa_1} &\rightarrow \vec{\beta} \otimes \vec{\kappa_1} \\
\vec{\alpha}  \otimes \vec{\kappa_2} &\rightarrow \vec{\beta} \otimes \vec{\kappa_2}.
\end{align}
\end{subequations}
This is a good example of non-unicity of catalysts.
Indeed, for a given forbidden transition, $\ket{\alpha} \nrightarrow 
\ket{\beta}$, and fixed local dimensions for the catalysts, the set of 
allowed vectors $\vec{\kappa}$ is a polytope~\cite{SFDY05}.

Another interesting geometric fact is that catalysis is possible for every 
bipartite scenario, starting from $4\times 3$, 
\ie for all effective dimensions $m \geq n$, $m \geq 4$, and $n\geq 3$ it 
is possible to choose 
$\vec{\alpha}=(\alpha_{1},...,\alpha_{m})\nleftrightarrow(\beta_{1},...,\beta_{n})=\vec{\beta}$ with
catalyst $\vec{\kappa}$.
The essential step, after the previously described examples, is that given 
$\vec{\alpha}$, $\vec{\beta}$, and $\vec{\kappa}$, we generically can 
increase dimensions by one  and construct  
a forbidden transition  $\vec{\alpha'} \nrightarrow \vec{\beta'}$  with 
the same catalyst $\vec{\kappa}$, by using
$\vec{\alpha'}=(\alpha_{1},...,\alpha_{n}-\epsilon,\epsilon)$ and
$\vec{\beta'}=(\beta_{1},...,\beta_{n}-\epsilon',\epsilon')$
with small enough $\epsilon>0$ and $\epsilon' \geq 0$. 
This allow us to construct examples for all such $m$ and $n$.


\section{Self-Catalysis}\label{selfie}

In this section we address the main question of this paper for the case of LOCC 
convertibility: can a bipartite quantum
state be itself the catalyst of a forbidden conversion? 
To be more precise, is there  a forbidden
conversion $\ket{\alpha}\nrightarrow \ket{\beta}$, such that 
$\ket{\alpha}\otimes
\ket{\alpha}\rightarrow \ket{\beta}\otimes \ket{\alpha}$? 
Refining a little bit more, is it possible that we still have
$\ket{\alpha}\otimes
\ket{\alpha}\nrightarrow \ket{\beta}\otimes \ket{\alpha}$, but a larger 
number of copies of
$\ket{\alpha}$ would do the job, \ie  $\ket{\alpha}\otimes
\ket{\alpha}^{\otimes N}\rightarrow \ket{\beta}\otimes \ket{\alpha}^{\otimes N}$ for some
$N > 1$?

In Table~\ref{Table} we list examples to answer affirmatively these 
questions.
The first is an example of self-catalysis from a two-ququart state to a two-qutrit state.
Then, for the same scenario, there is a list of multi-copy self-catalysis, with the corresponding minimal number of copies.   
At the last row, we come back to single-copy self-catalysis from a two-ququart state to another two-ququart state.

\textit{Remark}. When $\ket{\alpha}$ has 4 non-zero Schmidt 
coefficients and $\ket{\beta}$ has 3 non-zero Schmidt coefficients, the phenomenon of 
self-catalysis can be noted, that is, even in the minimal dimensions to occur catalysis 
(see \cite{daniel}) the self-catalysis can also carry out. 
Geometrically, this mean that even is the smallest dimension (most 
restrictive) case, there are cases where the source state $\vec{\alpha}$ 
belongs to the polytope of catalysts for the reaction $\vec{\alpha} 
\nrightarrow \vec{\beta}$. 
Moreover, it is possible to construct
examples for any scenario through a similar argument presented at the Section~\ref{catalisys}.


\begin{table}
\begin{tabular}{c| c |c}
\hline
$\vec{\alpha}$ & $\vec{\beta}$ & $\sharp$Copies ($N$) \\
\hline
$({0.900},{0.081},{0.010},{0.009})$ & ({0.950},{0.030},{0.020}) & $N=1$ \\
\hline
$({0.900},{0.088},{0.006},{0.006})$ & & $N=2$ \\
$({0.908},{0.080},{0.006},{0.006})$ & & $N=3$\\
$({0.918},{0.070},{0.006},{0.006})$ & 
$({0.950},{0.030},{0.020})$  & $N=4$\\
$({0.925},{0.063},{0.006},{0.006})$ & & $N=5$\\
$({0.928},{0.060},{0.006},{0.006})$  &  & $N=6$ \\
\hline
$({0.900},{0.081},{0.010},{0.009})$ & ({0.950},{0.030},{0.019},{0.001}) & $N=1$ \\
\hline

\end{tabular}
\caption{\label{Table}Number $N$ of copies required to make $\ket{\alpha}^{\otimes N}$ a deterministic
catalyst for the process $\ket{\alpha}\rightarrow \ket{\beta}$.}
\end{table}

\subsection{Stability under small perturbations}\label{stab}

It is important to mention that the phenomenon of self-catalysis, as it happens 
with
catalysis, is generically robust against small perturbations of the state vectors involved. That 
is, suppose that one is aiming to perform a self-catalytic process $\ket{\alpha}\otimes 
\ket{\alpha}\rightarrow \ket{\beta}\otimes \ket{\alpha}$, but it actually implements states
$\ket{\alpha'}$, $\ket{\beta'}$, where $\ket{\alpha'}\approx \ket{\alpha}$ and 
$\ket{\beta'}\approx \ket{\beta}$. 
Our claim is: \emph{generically, if $\ket{\alpha}$ $\alpha$-access 
$\ket{\beta}$, then $\ket{\alpha'}$ $\alpha'$-access $\beta'$}.  
It is easy to see 
that this will be true, depending only on the inequalities implying that
$\ket{\alpha}\nrightarrow \ket{\beta}$, as well as those assuring that 
$\ket{\alpha}\otimes
\ket{\alpha}\rightarrow \ket{\beta}\otimes\ket{\alpha}$,  all be strict (except the last one, which is granted by normalisation). 
Denote by $|\vec{\lambda}|$
some norm (\eg Euclidean) of the vector $\vec{\lambda}$ and $n,m$ the sizes 
of Schmidt vectors
$\vec{\alpha},\vec{\beta}$. 
Now, assuming that we have, for
some $l$, $\sum_{j=1}^{k}\alpha_{j}<\sum_{j=1}^{k}\beta_{j}$ for $1\leq k\leq l$ and
$\sum_{j=1}^{k}\alpha_{j}>\sum_{j=1}^{k}\beta_{j}$ for $l< k\leq n-1$, the same set of inequalities
will hold for the entries of vectors $\vec{\alpha'}$ and $\vec{\beta'}$ if
$|\vec{\alpha}-\vec{\alpha}'|,|\vec{\beta}-\vec{\beta'}|<\epsilon$, as long as
$\epsilon$ is small enough. 
This implies that $\ket{\alpha'}$ and $\ket{\beta'}$ are still incomparable. A
similar reasoning can be applied when the incomparability of vectors
$\vec{\alpha}$ and $\vec{\beta}$ is due to two or more changes of signs in the inequalities. 
In the same way,
$\sum_{j=1}^{k}(\alpha\otimes\alpha)^{\downarrow}_j<\sum_{j=1}^{k}(\beta\otimes\alpha)^{\downarrow}_j$ for $k<nm-1$,
imply 
$\sum_{j=1}^{k}(\alpha'\otimes\alpha')^{\downarrow}_j<\sum_{j=1}^{k}(\beta'\otimes\alpha')^{\downarrow}_j$, 
if $|\vec{\alpha}-\vec{\alpha}'|,|\vec{\beta}-\vec{\beta}'|<\epsilon$,
for small enough $\epsilon$, which proves the claim.
Naturally, the argument includes the well-motivated situation when $\ket{\beta'} = \ket{\beta}$ as a special case.

\subsection{Self-Catalysis under LOCC for random Schmidt vectors}\label{tip}

We have numerically investigated how usual the phenomenon of 
self-catalysis among pairs of 
incomparable bipartite states is. 
Fixing the sizes of $\vec{\alpha}$ and $\vec{\beta}$, we 
randomly sample pairs of such vectors until we 
find incomparable ones. 
The sampling of each 
vector is done by uniformly sorting unitary vectors in $\mathbb{C}^n\otimes \mathbb{C}^n$, \textit{i.e.} sorting according 
to 
the \emph{Haar Measure} in the respective state spaces~\cite{Loyd,Zyc}, and 
then calculating the correspondent Schmidt vector.
After finding a pair of incomparable states, we test 
whether the first of the vectors can be used as a catalyst for the 
conversion. 
The results show that for this method of sampling and for small dimensional systems, the phenomenon is actually \emph{atypical}.  
Moreover, the numerical estimations seem to imply that the phenomenon of 
self-catalysis is atypical in any dimension.     

\begin{figure}[h] 

\includegraphics[scale=0.5]{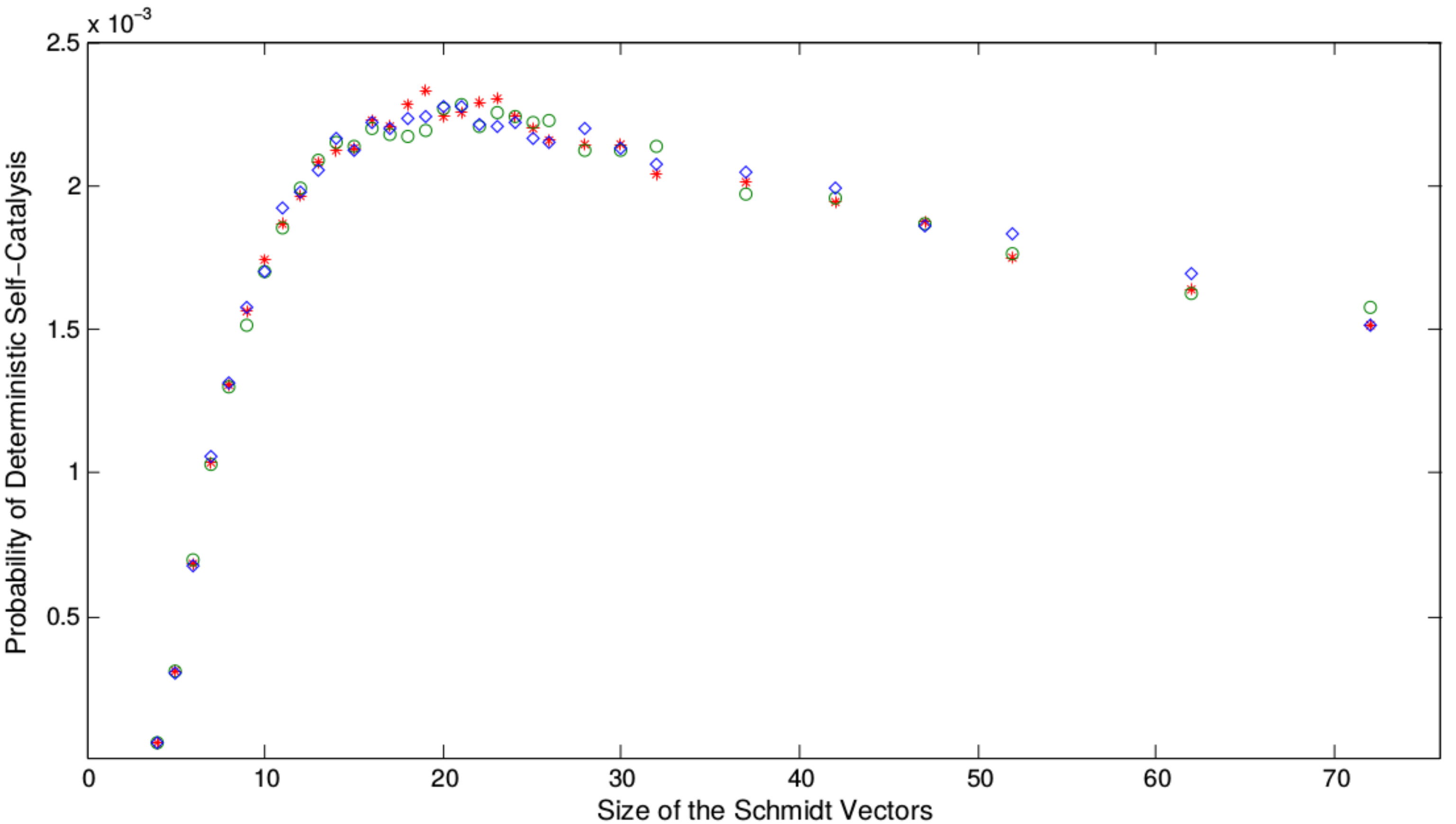}
\caption{Probability of finding a pair of states exhibiting self-catalysis as function of
the dimension of each system, that is 
$\mathbb{P}( \ket{\alpha}\otimes\ket{\alpha} \rightarrow 
\ket{\beta}\otimes \ket{\alpha} \,\, | \,\,  \ket{\alpha} \nleftrightarrow 
\ket{\beta} )$. Each symbol means an average over a distinct 
set of random choices. For each size explored there are three different 
symbols, which indicates a reasonable stability in this sampling 
process.}\label{cris2haar}
\end{figure}

Figure~\ref{cris2haar} shows a numerical estimation for $\mathbb{P}( 
\ket{\alpha}\otimes\ket{\alpha} \rightarrow 
\ket{\beta}\otimes \ket{\alpha} \,\, \arrowvert \,\,  \ket{\alpha} 
\nleftrightarrow 
\ket{\beta} )$, that is, the conditional 
probability of self-catalysis given a pair of incomparable Schmidt vectors, 
 as a function of their sizes. 
Note that this probability increases with the size of the Schmidt vectors,  
until sizes about $20$ and then starts to slowly decrease, but 
in fact for any dimension its order of magnitude 
shows that the phenomenon is present, but rare.

We believe the global maximum 
present in the Figure~\ref{cris2haar} can be explained by the algebraic 
character of the comparison between vectors and the concentration of measure phenomenon. In order to
$\ket{\alpha}\otimes \ket{\alpha}$ access 
$\ket{\alpha}\otimes \ket{\beta}$, all the corresponding inequalities of Eq.~\eqref{majoracao} must be satisfied. Then, in one hand, by
increasing the dimension of the vectors, this 
comparison become harder to be valid (more inequalities have to be obeyed). On the other hand,
it is known that a measure concentration 
phenomenon takes place for increasing dimensionality, in the sense that the Schmidt vectors become
typically closer and closer to the constant vector $(\frac{1}{d},...,\frac{1}{d})$. This phenomenon can be
highlighted by the average normalized entropy of the sorted vectors: if this average is close to $1$, it means that
most vectors are close to the constant vector. This average value is know to be exactly
$\frac{1}{\ln{d}}[\sum_{k=d+1}^{d^{2}}\frac{1}{k}-\frac{d-1}{2d}]$, as conjectured by
Page~\cite{page} and latter
proved by Foong and Kanno~\cite{foong}

If the vectors of the pair are incomparable, that is, some inequalities of Eq. \eqref{majoracao} are not satisfied, but
both are close to the constant vector, it will be easier for the catalyst to make the transition
possible, since the corresponding inequalities for the vectors with the catalyst attached will be
easier to be satisfied.

We can note that the concentration of measure increases very fast at low dimensions (see
Figure~\ref{entropy_haar}), helping the possibility of  ``organising'' the Schmidt 
vectors with the catalyst attached, which justifies
the corresponding increasing in the probability for self-catalysis. 
However, around dimension $20$ the
concentration happens much slower and presumably is not fast enough to overcome the size effect (which makes majorisation more
difficult), so the probability of finding a self-catalytic pair decreases after this point.
\begin{figure} 
\includegraphics[scale=0.5]{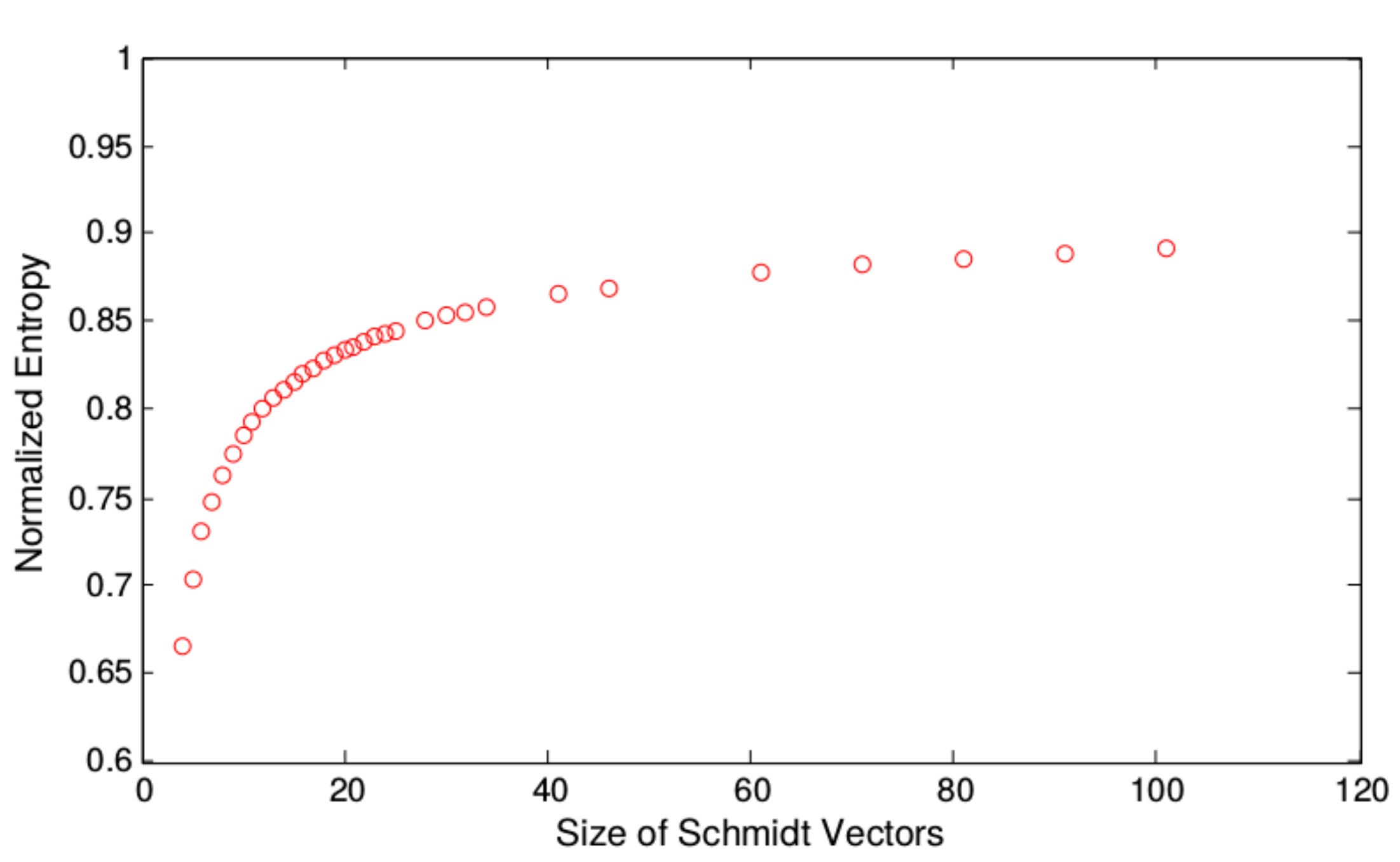}
\caption{Mean normalized entropy for Schmidt vectors 
sorted by Haar measure. Here we have explored the behavior of the entropy 
until effective dimension 100.}\label{entropy_haar}
\end{figure}

Summing up, the probability seems to be a smooth function
of the dimension, reaching its higher value for dimensions of each 
Schmidt vector close to $20$ and, apparently, converging to a  value 
considerably 
smaller than $1$. 

\begin{remark} We have restricted the analysis above to events where 
$\ket{\alpha} 
\nleftrightarrow \ket{\beta}$. 
For large dimensions we have checked that this event 
\emph{is} typical. 
Indeed, in such regime the entries of vectors 
$\vec{\alpha}$ and $\vec{\beta}$, before ordering, are essentially 
concentrated random variables fluctuating around a fixed value. If we look, 
then, to Eq.~\eqref{majoracao} we see that there 
is a good chance for the sign of the inequality to change as we vary the index 
$k$, implying that the states are incomparable. Therefore, we expect that 
$\mathbb{P}(\ket{\alpha} \nleftrightarrow \ket{\beta}) 
\approx 1$, 
for $n \gg 1$.\label{remark4}
\end{remark}


\section{Catalysis under SLOCC}\label{catslocc}

It is possible to generalize the concept of accessibility if we allow for 
non-deterministic
processes. 
In this case, we can look for the probability 
$P_{SLOCC}(\ket{\alpha}\rightarrow
\ket{\beta})$, or $P_{S}(\ket{\alpha}\rightarrow
\ket{\beta})$ for short, of having the state conversion 
$\ket{\alpha}\rightarrow\ket{\beta}$ under the best
local strategy, \ie optimizing $P$ under the conditions defining SLOCC. 
It is interesting to recall 
that the famous result on inconvertibility between W and GHZ states  
refers 
to such conditions on the multipartite scenario~\cite{GHZW}.
%
%

As shown by Vidal~\cite{vidal}, for the bipartite case, the Schmidt vectors 
also encode this 
maximal probability of conversion, $P_{S}(\ket{\alpha}\rightarrow 
\ket{\beta})$ through
\begin{theorem}
\label{thm-Vidal}
Let $\vec{\alpha}=(\alpha_{1},...,\alpha_{n})$ and 
$\vec{\beta}=(\beta_{1},...,\beta_{n})$ 
be ordered Schmidt vectors for states $\ket{\alpha}$ and $\ket{\beta}$, 
assuming $\alpha_n,\beta_n >0$. Define $\displaystyle{E_k(\lambda) = 
1-\sum_{l=1}^{k-1}\lambda_{l}}$. Then, the optimal transformation 
probability is given by
\begin{equation}\label{formulavidal}
P_{S}(\ket{\alpha}\rightarrow
\ket{\beta})=\min_{1\leq k\leq n}\left\lbrace 
\frac{E_k(\alpha)}{E_k(\beta)} \right\rbrace .
\end{equation}
\end{theorem}

For instance, if 
\begin{equation}\label{exampleslocc}
\vec{\alpha}=(0.6,0.2,0.2)\text{ and } 
\vec{\beta}=(0.5,0.4,0.1)
\end{equation}
are the Schmidt vectors for the two-qutrit states $\ket{\alpha}$ and 
$\ket{\beta}$, respectively, we get:

\begin{subequations}
\begin{align}
P_{S}(\ket{\alpha}\rightarrow \ket{\beta})=0.8, \\
P_{S}(\ket{\beta}\rightarrow \ket{\alpha})=0.5.
\end{align}
\end{subequations}
The following Proposition shows that the optimal probability of conversion 
$P_S$ attains $1$ precisely when $\ket{\alpha}$ access $\ket{\beta}$.

\begin{proposition} \label{conexao}
 Let $\vec{\alpha}$ and $\vec{\beta}$ be a pair of random independent 
Schmidt vectors with same size $n$, then the event $ \{ \ket{\alpha} 
\rightarrow \ket{\beta} \}$ is equal to event $\{P_{S}(\ket{\alpha} 
\rightarrow \ket{\beta})=1\}$. In particular $\mathbb{P} ( 
\ket{\alpha} 
\rightarrow \ket{\beta}) = \mathbb{P}(P_{S}(\ket{\alpha} 
\rightarrow 
\ket{\beta})=1)$.
\end{proposition}

\begin{proof}
Suppose that $\ket{\alpha} \rightarrow \ket{\beta}$, thus:
\begin{equation}
\sum_{i=1}^{k} \alpha_{i} \leq \sum_{i=1}^{k} \beta_{i}, \,\, \forall k 
\,\, \in \{1,2,...,n \}. 
\end{equation}
Then $E_{k}(\alpha)=1-\sum_{i=1}^{k-1}\alpha_{k} \geq 
1-\sum_{i=1}^{k-1}\beta_{k}=E_{k}(\beta), \,\, \forall \,\, k \in 
\{1,2,..., n-1 \}$ with the equality holding if $k=1$, therefore 
$P_{S}(\ket{\alpha} \rightarrow \ket{\beta})=1.$ The converse is similar.
\end{proof}

Note that from our considerations at the end of section~\ref{tip} we expect 
that $\mathbb{P} ( \ket{\alpha} \rightarrow \ket{\beta}) \approx 0$ (using LOCC) for 
large $n$, since the event $\{ \ket{\alpha} \rightarrow \ket{\beta} \}$ is 
in the complement of $\{ \ket{\alpha} \nleftrightarrow 
\ket{\beta}\}$.

We numerically estimate the SLOCC average rate of conversion between random 
states.
For this, we generate
incomparable Schmidt vectors following the \emph{Haar measure} and 
calculate the probabilities of conversion from the first to the
second and also the maximum conversion probability.
The results shown in Figure~\ref{Fig:ConvRatehaar} clearly shows that, 
given a 
random pair $\alpha$,
$\beta$, it is very common to have a large probability of conversion from 
some of them to the other,
 that is $\mathbb{E}[\text{max}\{P_{S}(\ket{\alpha} \rightarrow 
\ket{\beta},P_{S}(\ket{\beta} \rightarrow \ket{\alpha})\}] \gtrsim 0.8$. 
Moreover we have a smaller, but still significant, average probability 
of 
conversion $\mathbb{E}[P_{S}(\ket{\alpha} \rightarrow 
\ket{\beta})]$, slightly below $0.6$.

\begin{figure}[h]
\centering
\includegraphics[scale=0.5]{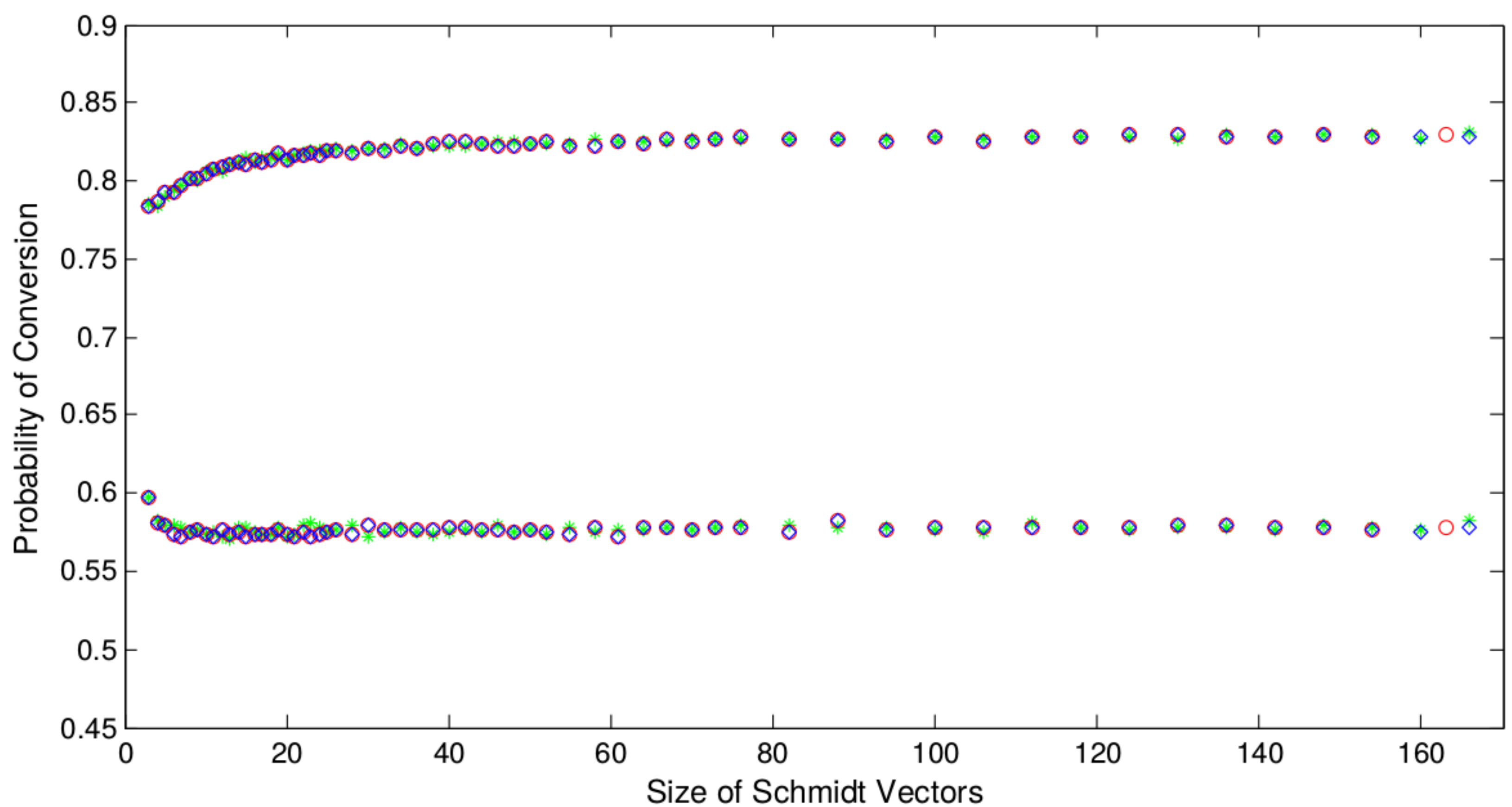}
\caption{Sampled probability of conversion under SLOCC for randomly chosen 
(following the Haar measure) 
incomparable states as function 
of the dimension of each system. Each symbol represents an average over a distinct set of randomly
chosen pairs. There are two sets of three symbols for each explored 
dimension. For each dimension,
the smaller results consider conversion from the first to the second 
$\mathbb{E}[P_{S}(\ket{\alpha} \rightarrow 
\ket{\beta})]$, and 
the larger the maximum
conversion rate $\mathbb{E}[\text{max}\{P_{S}(\ket{\alpha} \rightarrow 
\ket{\beta},P_{S}(\ket{\beta} \rightarrow \ket{\alpha})\}]$.}
\label{Fig:ConvRatehaar}
\end{figure}

Also in the probabilistic scenario the presence of an extra state can 
improve the probability of conversion between two states \cite{FDY04}. As a chemical catalyst, this extra state
is used, but not consumed, to increase the rate (probability) of a reaction (conversion).  In the
above example, if 
$\vec{\kappa}=(0.65,0.35)$, we arrive at:
\begin{subequations}
\begin{align}
\vec{\alpha} \otimes \vec{\kappa}&=(0.39, 0.21, 0.13, 0.13, 0.07, 0.07), \\
\vec{\beta} \otimes \vec{\kappa} &= (0.325, 0.26, 0.175, 0.14, 0.065, 0.035).
\end{align}
\end{subequations}
Therefore $P_{S}(\ket{\alpha} \otimes \ket{\kappa} \rightarrow \ket{\beta} 
\otimes \ket{\kappa}) \simeq 
0.904$, and $\ket{\kappa}$ can be viewed as a \emph{probabilistic-catalyst} 
in the stochastic scenario for the conversion that starts in $\alpha$ and 
ends in $\beta$, despite the fact that $P_{S}(\ket{\beta} \otimes 
\ket{\kappa} \rightarrow \ket{\alpha} \otimes \ket{\kappa}) =
0.5$, and then $\ket{\kappa}$ does not increases the probability of  
conversion for the transformation that starts in $\beta$ and ends in 
$\alpha$. In this sense
Jonathan and Plenio pointed out~\cite{daniel} that if $P_{S}(\ket{\alpha} 
\rightarrow \ket{\beta})$, 
under the best local strategy, is equal to $\alpha_n / \beta_n$, then 
this probability can 
not be increased by the presence of any catalyst state. Feng 
\textit{et.~al.}~\cite{FDY04} 
improved 
this result, obtaining the following theorem:
\begin{theorem}\label{thm_orientais_prob}
 Let $\vec{\alpha}$ and $\vec{\beta}$ be two $n$--dimensional probability vectors written in non-increasing order. 
 There is a probability vector $\vec{\kappa}$ such that $P_{S}(\vec{\alpha} 
\otimes \vec{\kappa} \rightarrow \vec{\beta} \otimes \vec{\kappa}) > 
P_{S}(\vec{\alpha} \rightarrow 
\vec{\beta})$ if, and only if, 

\begin{equation}
P_{S}(\vec{\alpha} \rightarrow \vec{\beta}) <
\text{\normalfont{min}}\left\lbrace 
\frac{\alpha_n}{\beta_n},1 \right\rbrace .
\end{equation}
\end{theorem}

A deeper connection between LOCC catalysis and its stochastic counterpart   
connects the probability of occurrence of the event 
$ \{ \alpha \rightarrow \beta \}$ and the maximal probability of stochastic 
conversion, $P_{S}(\ket{\alpha}\otimes \ket{\phi}\rightarrow 
\ket{\beta} \otimes \ket{\phi})$.

\begin{proposition} \label{boundlemma}
 Let $\vec{\alpha}$ and $\vec{\beta}$ be a pair of random independent 
Schmidt vectors with same size $n$. Then
\begin{equation}
\mathbb{P}[\sup_{\phi} \{P_{S}(\ket{\alpha} \otimes 
\ket{\phi} \rightarrow 
\ket{\beta} \otimes \ket{\phi}) \} > P_{S}(\ket{\alpha} \rightarrow 
\ket{\beta}) ] \quad \geq \quad  \frac{1}{2} - \mathbb{P}(\ket{\alpha} 
 \rightarrow \ket{\beta}).
\end{equation}
\end{proposition}

\begin{proof}

\begin{subequations}
\begin{align}
 \mathbb{P}[P_{S}(\ket{\alpha} \otimes \ket{\phi} \rightarrow 
\ket{\beta} \otimes \ket{\phi}) > P_{S}(\ket{\alpha} \rightarrow 
\ket{\beta}) ] & = 1 - \mathbb{P} \{  P_{S}(\ket{\alpha} \rightarrow 
\ket{\beta}) = \text{min}(\alpha_{n}/ \beta_{n},1) \} \label{l1} \\ 
& \geq 1  - 
\mathbb{P}\{ P_{S}(\ket{\alpha} \rightarrow \ket{\beta}) = \alpha_n / 
\beta_{n} \}
- \mathbb{P} \{ \ket{\alpha} \rightarrow \ket{\beta} \}  \label{l2} \\
& \geq \frac{1}{2} - \mathbb{P} \{ \ket{\alpha} \rightarrow \ket{\beta}  
\}. \label{l3}
\end{align}
\end{subequations}
Where~\eqref{l1} comes from Thm~\ref{thm_orientais_prob}, \eqref{l2} from 
set theory and Proposition \ref{conexao}, and finally \eqref{l3} from 
$\{P_{S}(\ket{\alpha} \rightarrow 
\ket{\beta}) = \alpha_n / 
\beta_{n} \} \subseteq \{ \alpha_n \leq \beta_{n}\} 
$.
\end{proof}
\begin{remark}
   From Remark~\ref{remark4} we know that as $n$ grows, $ \mathbb{P}(\ket{\alpha} 
 \rightarrow \ket{\beta})\approx 0$, so we conclude from Proposition~\ref{boundlemma} that $\mathbb{P}[P_{S}(\vec{\alpha} \rightarrow \vec{\beta}) <
\text{\normalfont{min}}\left\lbrace 
\frac{\alpha_n}{\beta_n},1 \right\rbrace]\gtrsim 1/2$. 
\label{remark2}
\end{remark}

In Ref.~\cite{FDYIEEE} a necessary and sufficient
condition for a state $\ket{\kappa}$ to work as a probabilistic catalyst was provided:
\begin{theorem}\label{theorem_feng}
 Suppose that $\vec{\alpha}$ and $\vec{\beta}$ are two non-increasingly ordered  
 $n$--dimensional probability vectors, and $P(\vec{\alpha} \rightarrow \vec{\beta}) 
 < \min\left\lbrace\frac{\alpha_n}{\beta_n},1\right\rbrace$. Define
 \begin{equation} \label{eq_oriente}
 L= \left\lbrace l; 1<l<n, \,\, \mbox{and} \,\, P(\vec{\alpha} \rightarrow \vec{\beta}) = 
\frac{E_{l}(\vec{\alpha})}{E_{l}(\vec{\beta})}\right\rbrace.
 \end{equation}
 Then a non-increasingly ordered $k$-dimensional probability vector $\vec{\kappa}$ serves
 as a probabilistic catalyst for the conversion from $\ket{\alpha}$ to $\ket{\beta}$ if, and 
 only if, for all $r_1, r_2, ..., r_k \in L \cup\{n+1\}$ satisfying $r_1 \geq r_2 \geq ...
  \geq r_k \neq n+1$, there exist $i$ and $j$, with $1 \leq j < i \leq k$, such that
  \begin{align}\label{eq_desigualdades}
 \frac{\kappa_i}{\kappa_j} < \frac{\beta_{r_j}}{\beta_{r_{i}-1}} \,\, \mbox{or} \,\, 
 \frac{\kappa_i}{\kappa_j} > \frac{\beta_{r_{j}-1}}{\beta_{r_{i}}}.
\end{align}
By definition, whenever  one of the 
inequalities~(\ref{eq_desigualdades}) includes an index  $n+1$, it is considered to be
violated, so the other one must necessarily be satisfied.
\end{theorem}
We should stress a couple of facts about the set $L$: first of all, note 
that it is a key ingredient for identifying catalysts for a given 
conversion, since in a certain sense it determines which indexes are 
really important for the comparison between 
$\vec{\beta}$ and $\vec{\kappa}$. Secondly observe that typically $L$ has 
only
one element $l$, \ie the minimum which determines the 
probability of 
conversion (see Thm. 2) is non-degenerate.

\section{Self-Catalysis under SLOCC}\label{sefslocc}

The phenomenon of self-catalysis can also take place when considering conversions under SLOCC.
Namely, if a conversion $\ket{\alpha}\rightarrow\ket{\beta}$ takes place with optimal probability
$0<p<1$, it can be the case that the optimal probability for
$\ket{\alpha}\otimes\ket{\alpha}\rightarrow\ket{\beta}\otimes\ket{\alpha}$ be $p' > p$.
Indeed, for the same Schmidt vectors $\vec{\alpha}$ and $\vec{\beta}$ given by Eq.~\eqref{exampleslocc}, there is a gain in the probability of conversion
if we use the state $\ket{\alpha}$ itself as a catalyst. Using Eq.~\eqref{formulavidal}, we have:
\begin{equation}
P_{S}(\ket{\alpha}\otimes\ket{\alpha}\rightarrow 
\ket{\beta}\otimes\ket{\alpha})\simeq
0.889>0.800 = P_{S}(\ket{\alpha}\rightarrow \ket{\beta}). 
\end{equation}

As it happens in the context of LOCC operations, the conversion between states can depend on the
number of attached copies of $\ket{\alpha}$. 
Table~\ref{table2} shows, for the example
we are considering (given by Eq.~\eqref{exampleslocc}), how the probability 
of 
conversion increases with the number of copies of $\ket{\alpha}$ to be used 
as catalyst.

This example may suggest that, by increasing the number of copies of $\ket{\alpha}$, the
probability of conversion approaches one. This is not always the case, however. Note that we
can bound from above the probability in Thm.~\ref{thm-Vidal} for the 
pair $\ket{\alpha}\otimes
\ket{\alpha}^{\otimes N}$ and $\ket{\beta}\otimes \ket{\alpha}^{\otimes N}$, 
since 
$P_{S}(\ket{\alpha}\otimes 
\ket{\alpha}^{\otimes
N}\rightarrow \ket{\beta}\otimes \ket{\alpha}^{\otimes N})\leq 
\alpha_{n}/\beta_{n},$
for all $N\geq 1$. Therefore, as long as 
$\alpha_{n}/\beta_{n}<1$, no matter how
many copies of $\ket{\alpha}$ we have, the probability of conversion will 
not exceed $\alpha_{n}/\beta_{n}$. The previous reasoning allows us to 
state the following Proposition:
\begin{proposition}
 Let $\vec{\alpha}$ and $\vec{\beta}$ be a pair of ordered Schmidt vectors 
with $n$ 
non-null components. If $\alpha_n / \beta_n < 1$, then 
\begin{equation}
 P_{S}(\ket{\alpha} \otimes \ket{\alpha}^{\otimes N} \rightarrow 
\ket{\beta} \otimes 
\ket{\alpha}^{\otimes N}) \leq \frac{\alpha_n}{\beta_n}, \,\, \forall \,\, 
N \geq 0.
\end{equation}
\end{proposition}

For example, given
$\vec{\alpha}=(0.60,0.21,0.10,0.09)$ and 
$\vec{\beta}=(0.55,0.25,0.10,0.10)$, 
 we have a pair of
incomparable states with   
$P_{S}(\ket{\alpha}\rightarrow
\ket{\beta})=0.88$, but since $\alpha_{4} / \beta_{4} = 0.9$,
the probability of conversion under SLOCC using self-catalysis is limited by 
$0.9$ and indeed, for this case, $N=1$ is already optimal, since 
$P_{S}(\ket{\alpha} \otimes \ket{\alpha} \rightarrow \ket{\beta} \otimes 
\ket{\alpha})=0.9$. Analogously, 
for the pair $\vec{\alpha^{\prime}}=(0.40,0.34,0.15,0.11)$ and 
$\vec{\beta^{\prime}}=(0.50,0.21,0.17,0.12)$, Table~\ref{table3} shows 
the behavior of $P_{S}(\ket{\alpha^{\prime}} \otimes 
\ket{\alpha^{\prime}}^{\otimes 
N} \rightarrow 
\ket{\beta^{\prime}} \otimes 
\ket{\alpha^{\prime}}^{\otimes N})$ with the number $N$ of copies of 
$\ket{\alpha^{\prime}}$, and since $\alpha^{\prime}_{n}/\beta^{\prime}_{n} 
= 0.91667 < 1$ the probability of conversion may increase, but can not 
reach 1. 
Moreover, there is the case (see 
Table~\ref{table4}) where, for a given pair $\vec{\alpha}, \vec{\beta}$ of 
incomparable Schmidt 
vectors, the probability of conversion increases
monotonically 
with respect to $N$ and, for an $N _{0}< \infty $,
$P_{S}(\ket{\alpha} \otimes 
\ket{\alpha}^{\otimes N_{0}} \rightarrow \ket{\beta} \otimes 
\ket{\alpha}^{\otimes N_{0}})=1$ and, by Proposition \ref{conexao}, $\ket{\alpha} \otimes 
\ket{\alpha}^{\otimes N_{0}} \rightarrow 
\ket{\beta} \otimes \ket{\alpha}^{\otimes N_{0}}$.

\begin{table}
\begin{center}
\begin{tabular}{c | c}
\hline
$\sharp$ Copies (N) & $P_{S}(\ket{\alpha}\otimes\ket{\alpha}^{\otimes N} 
\rightarrow \ket{\beta}
\otimes 
\ket{\alpha}^{\otimes N})$  \\

\hline
 0&  $\simeq$ 0.800\\
\hline
 1&  $\simeq$ 0.889\\
\hline
  2&  $\simeq$ 0.907 \\
\hline
 3&  $\simeq$ 0.926 \\
\hline
  4&  $\simeq$ 0.932 \\
\hline
 5&  $\simeq$ 0.940\\
\hline
 6&  $\simeq$ 0.944\\
\hline
 7&  $\simeq$ 0.947\\
\hline
 8&  $\simeq$ 0.950\\
\hline
 9&  $\simeq$ 0.952\\
\hline
10&  $\simeq$ 0.955\\
\end{tabular}
\caption{Increase of optimal probability for the conversion
$\ket{\alpha}\rightarrow\ket{\beta}$, where $\vec{\alpha}=(0.6,0.2,0.2)$ and
$\vec{\beta}=(0.5,0.4,0.1)$, with the number of copies $N$ of $\ket{\alpha}$ 
used as a
catalyst.}\label{table2}
\end{center}
\end{table}

\begin{table}
\begin{center}
\begin{tabular}{c | c}
\hline
$\sharp$ Copies (N) & $P_{S}(\ket{\alpha^{\prime}} \otimes
\ket{\alpha^{\prime}}^{\otimes N} 
\rightarrow 
\ket{\beta^{\prime}} \otimes 
\ket{\alpha^{\prime}}^{\otimes N})$  \\

\hline
 0&  $\simeq$ 0.8965\\
\hline
 1&  $\simeq$ 0.9038\\
\hline
  2&  $\simeq$ 0.9072 \\
\hline
 3&  $\simeq$ 0.9092 \\
\hline
  4&  $\simeq$ 0.9105 \\
\hline
 5&  $\simeq$ 0.9109\\
\hline
 6&  $\simeq$ 0.9110\\
\end{tabular}
\caption{Increase of optimal probability for the conversion
$\ket{\alpha^{\prime}}\rightarrow\ket{\beta^{\prime}}$, where 
$\vec{\alpha^{\prime}}=(0.40,0.34,0.15,0.11)$ and
$\vec{\beta}=(0.50,0.21,0.17,0.12)$, with the number of copies $N$ of 
$\ket{\alpha^{\prime}}$ used as a
catalyst.}\label{table3}
\end{center}
\end{table}

\begin{table}
\begin{center}
\begin{tabular}{c | c}
\hline
$\sharp$ Copies (N) & $P_{S}(\ket{\alpha}\otimes\ket{\alpha}^{\otimes N} 
\rightarrow \ket{\beta}
\otimes 
\ket{\alpha}^{\otimes N})$  \\

\hline
 0&  $\simeq$ 0.600\\
\hline
 1&  $\simeq$ 0.818\\
\hline
  2&  $\simeq$ 0.911 \\
\hline
 3&  $\simeq$ 0.957 \\
\hline
  4&  $\simeq$ 0.981 \\
\hline
 5&  $\simeq$ 0.994\\
\hline
 6&  $=$ 1\\
\hline
\end{tabular}
\caption{Increase of optimal probability for the conversion
$\ket{\alpha}\rightarrow\ket{\beta}$, where 
$\vec{\alpha}=(0.928,0.060,0.006,0.006)$ and
$\vec{\beta}=(0.950,0.030,0.0195,0.0005)$, with the number of copies $N$ of 
$\ket{\alpha}$ used as a
catalyst.}\label{table4}
\end{center}
\end{table}

A particular case of Theorem~\ref{theorem_feng} allows us to obtain a 
necessary and sufficient
condition for having probabilistic self-catalysis for a single copy:
 
 \noindent
 \textbf{Criterion}\emph{
  Let $\vec{\alpha}$ and $\vec{\beta}$ be two $n$--dimensional Schmidt vectors with 
$P_{S}(\vec{\alpha} \rightarrow \vec{\beta}) 
 < \min\left\lbrace \frac{\alpha_n}{\beta_n},1 \right\rbrace$ and
 \begin{equation}
 L= \left\lbrace l; 1<l<n, \,\, \mbox{and} \,\, P_{S}(\vec{\alpha} 
\rightarrow \vec{\beta}) = 
\frac{E_{l}(\vec{\alpha})}{E_{l}(\vec{\beta})}\right\rbrace .
 \end{equation}
The vector $\vec{\alpha}$ serves as a probabilistic self-catalyst for the transformation from 
$\ket{\alpha}$ to $\ket{\beta}$ if, and  only if, for all $r_1, r_2, ..., r_n \in L \cup\{n+1\}$ 
satisfying $r_1 \geq r_2 \geq ... \geq r_n \neq n+1$, there exist $i$ and $j$, with $1 \leq j <  i 
\leq n$, such that
%
\begin{align}
 \frac{\alpha_{i}}{\alpha_{j}} &< \frac{\beta_{r_j}}{\beta_{r_{i}-1}} 
\label{ineq5} \\ 
&\mbox{\text{or}}  \nonumber \\ 
 \frac{\alpha_{i}}{\alpha_{j}} &> \frac{\beta_{r_{j}-1}}{\beta_{r_{i}}}.\label{ineq4}
\end{align}
Whenever  one of the inequalities~\eqref{ineq5} or~\eqref{ineq4} 
has an index $n+1$, it is considered to be
violated, so the other one must necessarily be satisfied.}

\begin{remark}
 Again, we observe that typically $L$ has only one element $l$, since the minimum which
determines the probability of conversion is non-degenerate with probability $1$.
\label{remark3}
\end{remark}

\subsection{{Self-Catalysis under SLOCC for random Schmidt vectors}}\label{tip2}

The criterion above, together with Proposition~\ref{boundlemma} and the 
behaviour of $L$, have interesting consequences for the probability of 
occurrence of self-catalysis. 

From Remark~\ref{remark2}, we know that the event
$[P_{S}(\ket{\alpha}\rightarrow\ket{\beta})<\text{min}\{\frac{\alpha_{n}}{\beta{n}},1\}]$ has
probability $\gtrsim 1/2$. Conditioning on this event we can then analyze the validity of
Ineqs.~\eqref{ineq5} and \eqref{ineq4}. With probability $1$ we must have 
$L=\{l\}$, for some $1<l<n$. Therefore, we can choose the indexes
$r_{i}$ in only two ways: either $r_{1}=r_{2}=...=r_{n}=l$ or $r_{1}=r_{2}=...=r_{j}=n+1$ and $r_{j+1}=...=r_{n}=l$, for some $j$.
For the second case, one can always satisfy one of the inequalities using the index $n+1$. For the
first case, the r.h.s of the inequalities always have index $l$ and we can lower bound the
probability for at least one of Inequalities \eqref{ineq4} to be valid:
\begin{align}
\mathbb{P}[\max_{1\leq j<i\leq n}\{\frac{\alpha_{i}}{\alpha_{j}}\}>\frac{\beta_{l-1}}{\beta_{l}}]\geq
\mathbb{P}[\max_{1<j\leq n}\{\frac{\alpha_{j-1}}{\alpha_{j}}\}>\max_{1<j\leq n}\{\frac{\beta_{j-1}}{\beta_{j}}\}]=1/2, 
\end{align}
using that $\max_{1\leq j<i\leq n}\{\frac{\alpha_{i}}{\alpha_{j}}\}=\max_{1< j\leq
n}\{\frac{\alpha_{j-1}}{\alpha_{j}}\}$, since $\vec{\alpha}$ is ordered, and the fact
that the two random variables on the second term are independent and identically distributed. Putting these together
we get that the probability for having SLOCC self-catalysis is $\gtrsim \frac{1}{4}.$

\begin{figure}[htbf]
\centering
\includegraphics[scale=0.5]{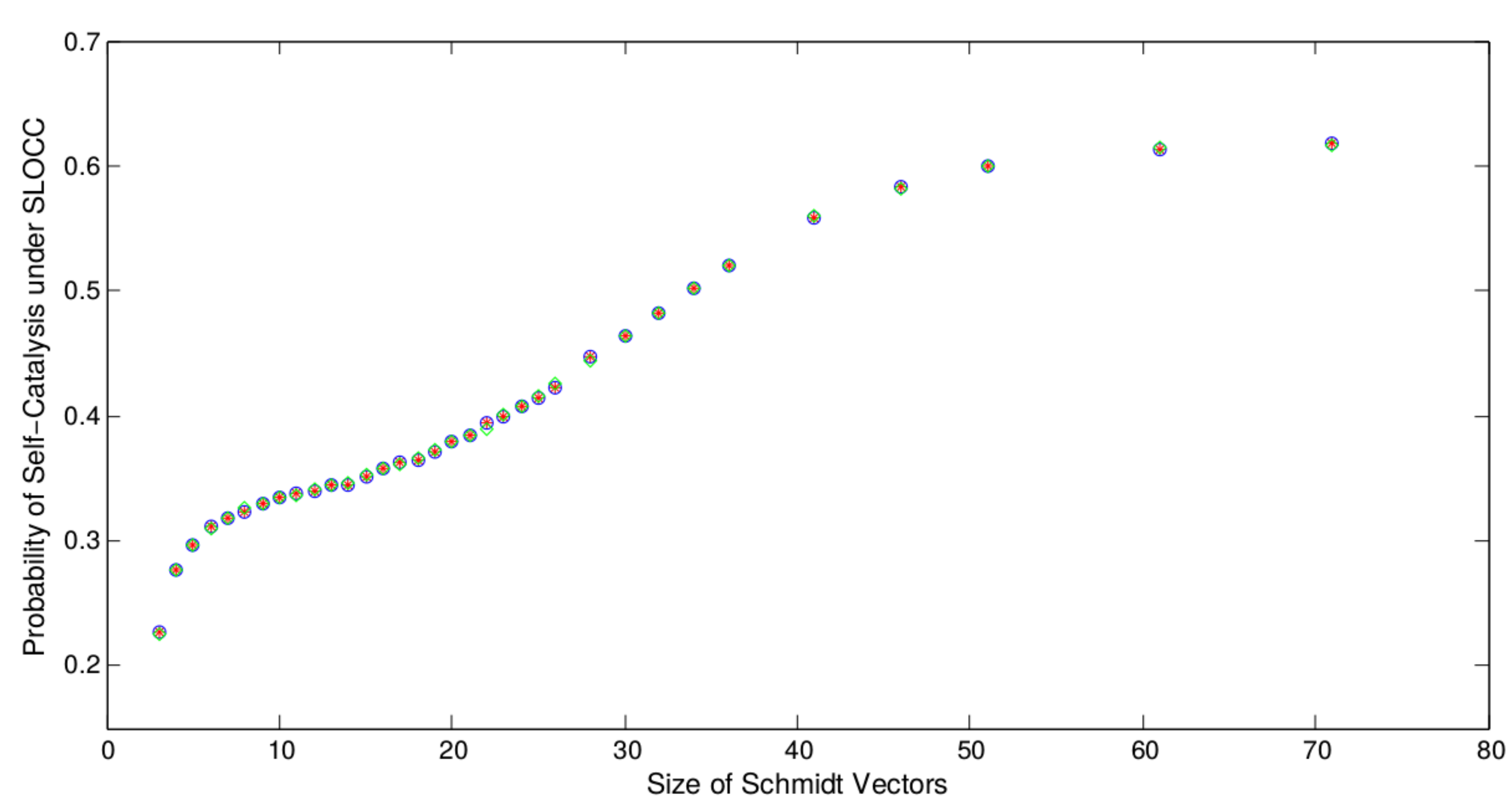}
\caption{Probability of finding a pair of states exhibiting self-catalysis 
under SLOCC, $\mathbb{P}[P_{S}(\ket{\alpha} \otimes \ket{\alpha} 
\rightarrow 
\ket{\beta} \otimes \ket{\alpha}) > \mathbb{P}(\ket{\alpha} \rightarrow 
\ket{\beta}) ]$, as function 
of the dimension of each system. Each symbol represents an average over a 
new set of randomly (Haar) chosen pairs. There are three symbols for each 
explored 
dimension.}
\label{cris3haar}
\end{figure}

We have also numerically investigated the typicality of probabilistic self-catalysis by 
1) sorting a pair of incomparable Schmidt vectors 
of same fixed dimension; 2) counting how many
of them show the effect; and 3) computing the average gain in probability. To 
be more specific, we
consider as a success case the situation where the pairs are such that
$p_{1}=P_{S}(\ket{\alpha}\rightarrow\ket{\beta})<P_{S}(\ket{\alpha}\otimes
\ket{\alpha}\rightarrow\ket{\beta}\otimes \ket{\alpha})=p_{2}$ and we compute the average value of
$p_{2}-p_{1}$. In order to avoid counting cases where $p_{1}<p_{2}$ due to numerical
fluctuations, we only consider as valid those vectors where $p_{2}>\left(1+10^{-5}\right)p_{1}$. 
The results are shown in Figures~\ref{cris3haar} and \ref{cris4haar}.

\begin{figure}[ht]
\centering
\includegraphics[scale=0.5]{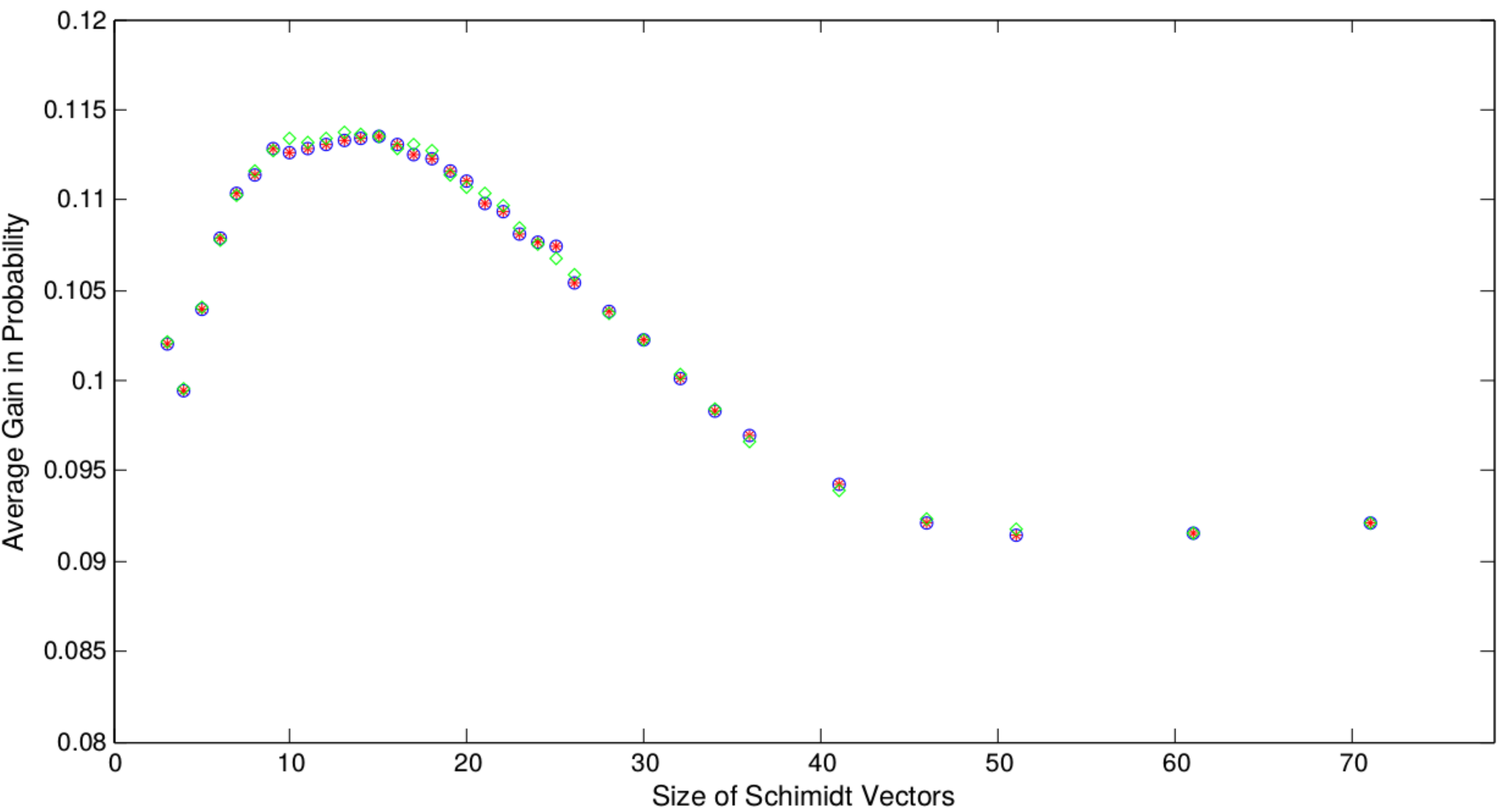}
\caption{Average gain of probability,  $\mathbb{E}[P_{S}(\ket{\alpha} 
\otimes \ket{\alpha} \rightarrow \ket{\beta} \otimes 
\ket{\alpha})-P_{S}(\ket{\alpha} \rightarrow \ket{\beta})]$, as a function 
of the size of Schmidt vectors (sampled following Haar measure), considered 
only those pairs where 
self-catalysis occurs, \ie $P_{S}(\ket{\alpha} 
\otimes \ket{\alpha} \rightarrow \ket{\beta} \otimes 
\ket{\alpha}) > P_{S}(\ket{\alpha} \rightarrow \ket{\beta})$.}
\label{cris4haar}
\end{figure}
First note that they are consisted with the lower bound of $1/4$ estimated before. Comparing with the
deterministic case, probabilistic self-catalysis is much more frequent, as expected. Even more, here we \emph{do not} have
the same qualitative behaviour: the probability of having a pair of incomparable states exhibiting
self-catalysis increases monotonically with the size of the Schmidt 
vectors 
and seems to be
converging to a  value around $0.6$. 
Meanwhile, this is not the behaviour of the average probability gain, shown in  
Figure~\ref{cris4haar}. 
The average gain in probability is
relatively small for all system sizes, and decreases even more for larger 
sizes. But note that we consider only one copy of $\ket{\alpha}$
attached, and while for many copies the average gain must be 
larger. It is important to recall Fig.~\ref{Fig:ConvRatehaar}, however, which tell
us that a pair $\left( \vec{\alpha}, \vec{\beta}\right)$ has, on average, a probability 
of direct conversion close to $0.6$, which naturally bounds the catalytic gain to about $0.4$.

Finally,  Figure~\ref{Fig:Dpagainstphaar} represents, for 
randomly chosen pairs of 
incomparable Schmidt vectors with size $n=45$, the self-catalytic 
probability gain \emph{versus} direct conversion rates. 
The diagonal straight line just represents saturation of probability.
Some concentration close to the horizontal axis is natural, representing the cases where self-catalysis does not happen.
However it is not clear why there is the bold concentrated 
cloud in red, where the majority of the pairs fit. 
It empirically means that the most typical situation for a pair of Schmidt vectors of the same size is to have a large probability of conversion and to have a considerable (but not maximal) self-catalytical gain.

\begin{figure}[h]
\centering
\includegraphics[scale=0.5]{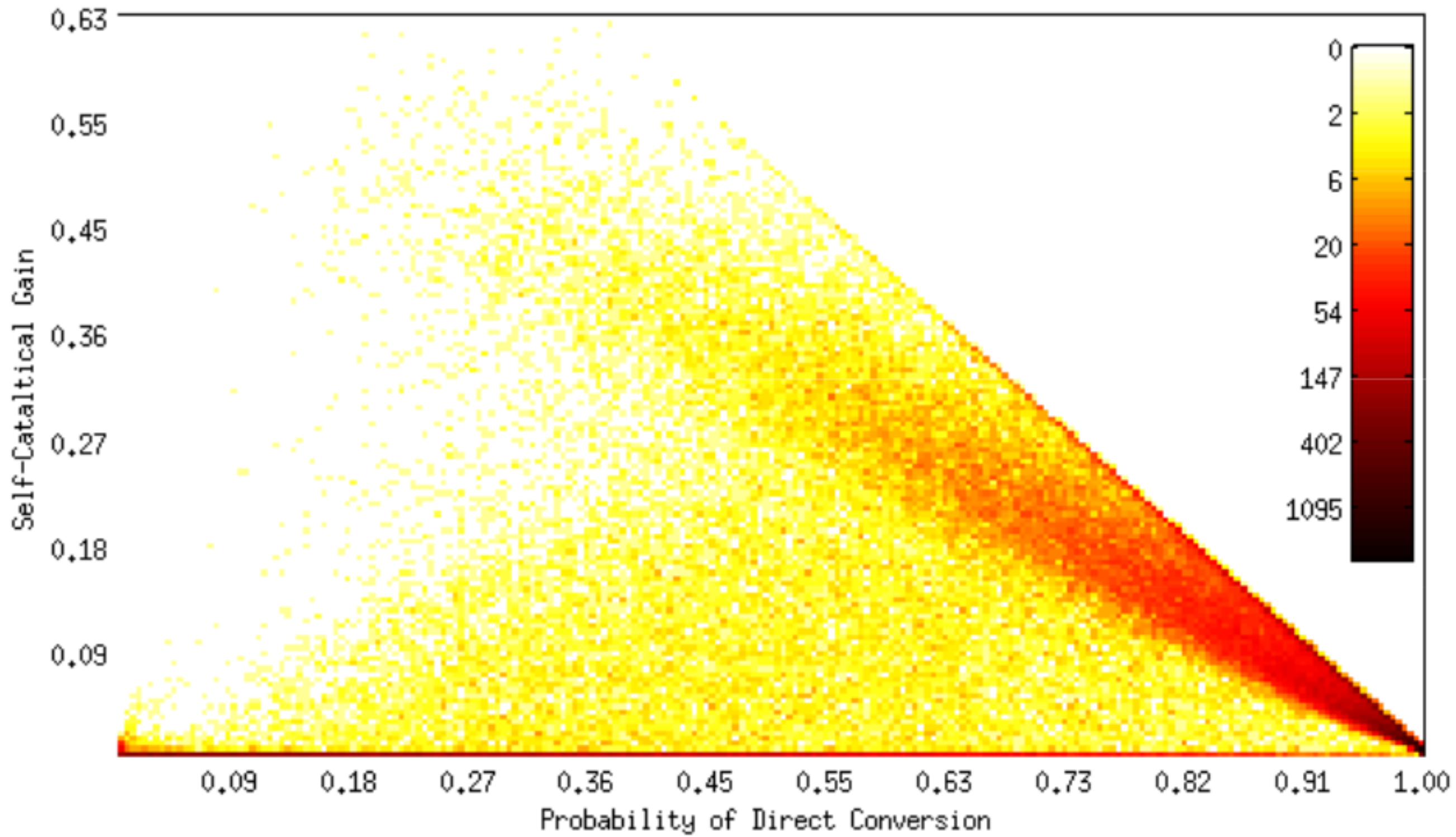}
\caption{Self-catalytic probability gain and direct conversion rates for 
randomly (Haar) chosen pairs of 
incomparable Schmidt vectors with size $n=45$. The colour represents the
number of pairs per pixel.}
\label{Fig:Dpagainstphaar}
\end{figure}
%
%

\section{Conclusions}\label{concl}

In this paper we have shown the possibility of self-catalytic entanglement 
conversion for
both LOCC and SLOCC scenarios by providing explicit examples of them. 
We have explored numerically and by arguments of typicality  how frequent they are by showing that it is much more 
common in SLOCC case
than the deterministic one, but despite the fact that self-catalysis 
under SLOCC become more common as the systems sizes increase, the direct 
self-catalysis decreases with the dimension. 
Moreover, we also investigated how the phenomenon may depend on the number of copies used as catalyst, finding examples of different behaviours, running from cases where there is no gain in considering multiple copies, to cases where the conversion becomes deterministic for a finite number  of copies.
Since we rest on numerical techniques, we could not guarantee the existence of a transition where the probability of conversion asymptotically goes to $1$. 
Finally, we computed the average gain in probability in the SLOCC case, showing that
this gain, as in the LOCC case, has a global maximum, and also decreases 
with larger systems sizes.

{In answering the question about existence of self-catalysis, we obtained 
many results, not
only on deterministic and probabilistic self-catalysis, but also on ordinary 
catalysis.
About the typicality of self-catalysis, our data support two conjectures: 
under LOCC, the
probability of finding a self-catalytic reaction increases monotonically 
attaining a global maximum for a dimension about $20$. 
We believe that the origin of this 
overall non-monotonic behaviour is due the competition between the 
\emph{sizes} of each Schmidt vectors and the measure \emph{concentration} 
phenomenon}. On the other hand, the data suggests that under SLOCC the probability of finding a self-catalytic reaction increases 
monotonically with the dimension. 

In a sense, we estimated numerically the volumes of the sets of pairs of 
Schmidt vectors where the
phenomena take place, but it was not possible to characterize completely the asymptotic behavior
with the vectors sizes. 
{In fact, some of our numerical results have a reasonable dependence on the way  we sort 
the random Schmidt vectors}.
That is something to be explored elsewhere. 
Perhaps a better understanding of 
the geometry of the sets involved could
help in that analysis.



\begin{acknowledgments}
We would like to thank Tha\'is Matos Ac\'acio, Jos\'e Am\^ancio,  Luiz Fernando Faria, Eduardo Mascarenhas,  Fernando de Melo, Roberto Oliveira,  and the anonymous referee for useful comments and stimulating discussions, as well as CNPq, CAPES and FAPEMIG for financial support.
\end{acknowledgments}



\end{document}